\documentclass{amsart}
\usepackage[utf8]{inputenc}
\usepackage[top=1in, bottom=1.5in, left=1.5in, right=1.5in]{geometry}

\usepackage[backend=biber,backref=true,style=alphabetic,giveninits=true]{biblatex} 

\usepackage{amsaddr}
\usepackage{amsmath}
\usepackage{amssymb}
\usepackage{mathtools}
\usepackage{pifont}
\usepackage{braket}
\usepackage{amsthm}
\usepackage{upgreek}
\usepackage{comment}
\usepackage{stmaryrd}
\usepackage{appendix}
\usepackage{bbm}
\usepackage{dsfont}
\usepackage{url}
\usepackage{tikz-cd}
\usepackage[shortlabels]{enumitem}
\usepackage{caption}
\usepackage{subcaption}
\usepackage{xcolor}
\usepackage{tgtermes}
\DeclareFontFamily{U}{dutchcal}{\skewchar\font=45 }
\DeclareFontShape{U}{dutchcal}{m}{n}{<-> s*[1.0] dutchcal-r}{}
\DeclareFontShape{U}{dutchcal}{b}{n}{<-> s*[1.0] dutchcal-b}{}
\DeclareMathAlphabet{\mathlcal}{U}{dutchcal}{m}{n}
\SetMathAlphabet{\mathlcal}{bold}{U}{dutchcal}{b}{n}
\newcommand{\ess}{\mathlcal{s}}

\usepackage[bookmarks=false,pdfstartview={FitH}]{hyperref}

\DeclarePairedDelimiter{\floor}{\lfloor}{\rfloor}

\newtheorem{theorem}{Theorem}[section]
\newtheorem{corollary}[theorem]{Corollary}
\newtheorem{lemma}[theorem]{Lemma}
\newtheorem{proposition}[theorem]{Proposition}

\newtheorem{remark}[theorem]{Remark}

\newcommand{\generate}[2]{\langle #1 \rangle_{\mathsf{#2}}}

\newcommand{\iu}{\mathrm{i}\mkern1mu}

\newcommand{\C}{\mathbb{C}}
\newcommand{\R}{\mathbb{R}}

\newcommand{\mf}[1]{\mathfrak{#1}}
\newcommand{\mc}[1]{\mathcal{#1}}

\newcommand{\diag}{\operatorname{diag}}

\newcommand{\conv}{\operatorname{conv}}

\newcommand{\SU}{\operatorname{SU}}
\newcommand{\SO}{\operatorname{SO}}

\newcommand{\U}{\operatorname{U}}

\renewcommand{\Im}{\operatorname{Im}}
\renewcommand{\Re}{\operatorname{Re}}

\newcommand{\Ad}{\operatorname{Ad}}
\newcommand{\ad}{\operatorname{ad}}
\newcommand{\tr}{\operatorname{tr}}

\newcommand{\id}{\mathds1}

\renewcommand{\epsilon}{\varepsilon}

\newcommand{\dmin}{{d_{\min}}}
\newcommand{\Uloc}{\U_{\mathrm{loc}}}

\newcommand{\uloc}{\mf u_{\mathrm{loc}}}

\begin{document}

\title[Optimal Control of Bipartite Systems]{Optimal Control of Bipartite Quantum Systems}

\author{Emanuel Malvetti {\small and} Léo Van Damme}

\address{School of Natural Sciences, Technische Universit\"at M\"unchen, 85737 Garching, Germany, and \\ Munich Center for Quantum Science and Technology (MCQST) \& Munich Quantum Valley (MQV)}

\begin{abstract}
Closed bipartite quantum systems subject to fast local unitary control are studied using quantum optimal control theory and a method of reduced control systems based on the Schmidt decomposition.
Particular focus is given to the time-optimal generation of maximally entangled states and product states, as well as to the problem of stabilizing quantum states with a certain amount of entanglement.
Explicit analytical solutions are given for general systems consisting of two qubits (as well as for bosonic and fermionic analogues) and also for a class of systems consisting of two coupled qutrits which is studied using the Pontryagin Maximum Principle.

\bigskip

\noindent\textbf{Keywords.} Bipartite entanglement, reduced control system, local unitary control, Schmidt decomposition, singular value decomposition, optimal control, Pontryagin Maximum Principle \medskip

\noindent\textbf{MSC Codes.} 
81Q93, 
15A18, 
81P42, 
49J15, 
49J24 

\end{abstract}

\maketitle

\section{Introduction}

Quantum entanglement is an important resource in pure-state quantum computation~\cite{Jozsa03,Vidal03,Animesh07}, quantum cryptography~\cite{BB84reprint,Gisin02}, and quantum sensing~\cite{Degen17}.
For the realization of such technologies it is necessary to be able to control the entanglement of a quantum system.
In particular the generation of sufficient entanglement as well as the stabilization of entangled states is essential.
In this paper we use quantum optimal control theory and a method of reduced control systems to address these challenges for closed bipartite quantum systems subject to fast local unitary control, focussing in particular on concrete low-dimensional systems.
 
The well-known Schmidt decomposition (equivalently the singular value decomposition) states that a pure quantum state of a bipartite system can always be transformed into a ``diagonal'' form using local unitary transformations.
The resulting diagonal elements are the singular values and they quantify the amount of entanglement between the two subsystems.
In a recent paper~\cite{BipartiteReduced24} we show that a closed bipartite quantum system with fast local unitary control can equivalently be described by a \emph{reduced control system} defined on the singular values of the state by factoring out the fast controllable degrees of freedom.
The general theory was established in~\cite{Reduced23} and applied to Markovian quantum systems with fast unitary control in~\cite{LindbladReduced23}.
Analogous results hold for bosonic and fermionic systems, where the Schmidt decomposition is to be replaced by the Autonne--Takagi and the Hua factorization respectively.

In this paper we apply optimal control theory to this reduced control system in some low dimensional cases.
This allows us, for instance, to find the optimal control sequence to prepare a maximally entangled state, or conversely to disentangle a given state.
Moreover, we address the problem of stabilizing states with a certain amount of entanglement.

\subsection{Outline}

Section~\ref{sec:control-systems} defines the full and the reduced control systems and recalls the equivalence between them.
The following sections use the reduced control system to derive time-optimal controls and stabilizing controls for the full control system in several concrete settings.
In particular, Section~\ref{sec:opt-ctrl-angle} treats those cases where the reduced state space is one-dimensional, for instance the case of two distinguishable qubits.
Finally, a higher-dimensional case of two distinguishable qutrits is studied in Section~\ref{sec:opt-ctrl-qutrits} using the Pontryagin Maximum Principle to derive time-optimal solutions for preparing maximally entangled states.

\section{Control Systems} \label{sec:control-systems}

We briefly recall the definitions of the control systems involved, focussing on the case of two distinguishable subsystems.
Further details and all relevant proofs, as well as the bosonic and fermionic cases, can be found in~\cite{BipartiteReduced24}.

\subsection{Full Control System} \label{sec:bilinear-system}

Let $\C^{d_1}$ and $\C^{d_2}$ be two finite dimensional Hilbert spaces of dimensions $d_1,d_2\geq2$.
The goal is to study the bilinear control system~\cite{Jurdjevic97,Elliott09} given by the following controlled Schrödinger equation:
\begin{equation} \label{eq:bilinear} \tag{\sf B}
\ket{\dot\psi(t)} = -\iu\Big(H_0 
+ \sum_{i=1}^{m_1} u_i(t) E_i\otimes\id 
+ \sum_{j=1}^{m_2} v_j(t) \id\otimes F_j
\Big)\ket{\psi(t)}\,,
\end{equation}
where $H_0\in\iu \mf{u}(d_1d_2)$ is the \emph{drift Hamiltonian} (or \emph{coupling Hamiltonian}), $E_i\in\iu\mf{u}(d_1)$ and $F_j\in\iu\mf{u}(d_2)$ are the \emph{(local) control Hamiltonians}, and $u_i$ and $v_j$ are the corresponding \emph{control functions}. 
We make the following key assumptions:
First, the control functions are locally integrable, in particular they may be unbounded.
Second, the control Hamiltonians generate the full local unitary Lie algebra\footnote{
It is possible to only consider the local special unitary Lie algebra. For details see~\cite[Rmk.~2.1, Sec.~2.4]{BipartiteReduced24}.
}:
$$\generate{\iu E_i\otimes\id,\id\otimes \iu F_j :\,i=1,\ldots,m_1,\,j=1,\ldots,m_2}{\mathsf{Lie}}=\uloc(d_1,d_2):=(\mf u(d_1)\otimes\id)\oplus(\id\otimes\,\mf u(d_2)).$$
Put simply, we have fast and full control over the group of local unitary transformations, which we denote $\Uloc(d_1,d_2):=\U(d_1)\otimes\U(d_2)$.

\subsection{Reduced Control System} \label{sec:reduced-system}


Due to the assumption of fast local unitary control, we can move arbitrarily quickly within the local unitary orbits of the system~\cite[Prop.~2.7]{Elliott09}.
The Schmidt decomposition (cf.~\cite[Sec.~9.2]{Bengtsson17})%
\footnote{This is equivalent to the complex singular value decomposition (see~\cite[Thm.~17.3]{Roman05}, a more basic exposition of the real case can be found in~\cite[Thm.~8.63]{Olver18}).}
states that using a local unitary transformation, a bipartite quantum state $\ket\psi\in\C^{d_1}\otimes\C^{d_2}$ can always be brought into the diagonal form $\sum_{i=1}^\dmin\sigma_i\ket i\otimes\ket i$, where $\sigma_i\in\R$ are the singular values of $\ket\psi$ (or rather the corresponding matrix) and $\dmin=\min(d_1,d_2)$.
The subspace of all such diagonal states is denoted $\Sigma\subset\C^{d_1}\otimes\C^{d_2}$.
Moreover, the singular values are well-defined up to order and signs, and two states have the same singular values if and only if they can be transformed into each other using local unitary transformations.
In particular, within the full control system~\eqref{eq:bilinear} two states are effectively equivalent if and only if they have the same singular values $\sigma_i$ (up to order and sign).
This strongly suggests that there should exist a ``reduced'' control system defined on the singular values.
This is indeed the case, as was shown in greater generality in~\cite{Reduced23} and for the present setting in~\cite{BipartiteReduced24}.

The vector $\sigma=(\sigma_i)_{i=1}^\dmin\in\R^\dmin$ containing the singular values of the quantum state will be the state of the reduced control system.
Due to the normalization of the state $\ket\psi$, it always holds that $\sum_{i=1}^\dmin\sigma_i^2=1$.
Hence the reduced state space is the so-called \emph{Schmidt sphere} $S^{\dmin-1}$.
We will see below that the dynamics on this sphere are given by rotations.
Moreover, the non-uniqueness of the singular values introduces a symmetry on the reduced state space.
Indeed, the \emph{Weyl group} acts on the space of singular values by permutations and sign flips, and partitions the space into a number of equivalent regions. 
By fixing the order and signs of the singular values (e.g.\ non-increasing order and non-negative signs) one can choose one of these regions, called the \emph{Weyl chamber}, to remove this ambiguity.


In order to switch between the set of bipartite quantum states and the Schmidt sphere we need the following maps:
The map $\diag:\R^\dmin\to\Sigma$ sends a vector of singular values to the corresponding diagonal quantum state via $\sigma\mapsto\ket{\sigma}:=\sum_{i=1}^\dmin\sigma_i\ket i\otimes\ket i$.
Conversely, the map $\Pi_\Sigma:\C^{d_1}\otimes\C^{d_2}\to\R^\dmin$ sends a state to its real diagonal via $\ket\psi\mapsto(\Re(\braket{ii|\psi}))_{i=1}^\dmin$.


To define the reduced control system on the Schmidt sphere $S^{\dmin-1}$, we have to understand the dynamics of the singular values of the state induced by the coupling Hamiltonian $H_0$. 
For this, consider a regular\footnote{Here \emph{regular} means that all singular values are distinct and non-zero.} and differentiable path of states $\ket{\psi(t)}=V(t)\otimes W(t)\ket{\sigma(t)}$ which is a solution to the full control system~\eqref{eq:bilinear}. 
Then it holds that $\dot\sigma(t)=-H_{V\otimes W}\sigma(t)$ where
$$
H_{V\otimes W} := \Pi_\Sigma\circ(V\otimes W)^* \iu H_0 (V\otimes W)\circ\diag
$$
is called an \emph{induced vector field} on the Schmidt sphere.
Here $\circ$ denotes function composition.
Indeed, it is easy to show that this vector field generates a rotation of the Schmidt sphere.
The collection of induced vector fields is denoted $\mf H := \{-H_U:U\in\Uloc(d_1,d_2)\}$.
These vector fields contain the possible dynamics of the reduced control system.
The following result gives a more concrete expression:

\begin{proposition} \label{prop:induced-vf}
Let $H_0=\sum_{k=1}^r A_k\otimes B_k$ denote an arbitrary coupling Hamiltonian. 
Then, for a local unitary $V\otimes W$, the induced vector field $-H_{V\otimes W}$ on $S^{\dmin}$ takes the form
$$
H_{V\otimes W}=\sum_{k=1}^m \Im(V^*A_kV\circ W^*B_kW).
$$
Here $\circ$ denotes the element-wise product (if $d_1\neq d_2$ the resulting matrix will have size $\dmin\times\dmin$).
\end{proposition}


With these definitions we can finally define the reduced control system as follows. Let $I$ be an interval of the form $[0,T]$ for $T\geq0$ or $[0,\infty)$. Then we say that $\sigma:I\to\R^\dmin$ is a solution to the reduced control system
\begin{equation} \label{eq:reduced} \tag{\sf R}
\dot\sigma(t)=-H_U\sigma(t), \quad \sigma(0)=\sigma_0\in S^{\dmin-1}
\end{equation}
if it is an absolutely continuous path 
satisfying~\eqref{eq:reduced} almost everywhere for some measurable control function $U:I\to\Uloc(d_1,d_2)$.
Note that such solutions always remain on the sphere $S^{\dmin-1}$.

\begin{remark} \label{rmk:different-reduced}
There are a few slightly different ways of defining the reduced control system which are given in~\cite{Reduced23}. 
The most intuitive, which is given here, is to consider the control system $\dot{\sigma}(t)=-H_{U(t)}(\sigma(t))$.
This is exactly equivalent to the differential inclusion $\dot\sigma(t)\in\mf H\sigma$, which is itself approximately equivalent to the ``relaxed'' differential inclusion $\dot\sigma(t)\in\conv(\mf H\sigma)$ where $\conv$ denotes the convex hull.
\end{remark}


The central result is the Equivalence Theorem~\cite[Thm.~2.10]{BipartiteReduced24} which states the precise sense in which the full and reduced control systems are equivalent. 
Informally, for every solution $\ket{\psi(t)}$ of the full control system, the corresponding singular values $\sigma(t)$ (when chosen non-negative and ordered non-increasingly) yield a solution to the reduced control system. 
Conversely, every solution to the reduced control system can be lifted approximately with arbitrarily small error to the full control system.


Importantly, under some assumptions, there is an exact and constructive way to compute such a lifted solution.
The relevant concept is that of the \emph{compensating Hamiltonian} $H_c$ for the state $\ket{\psi}$ defined by
$H_c=\ad_\psi^+(iH_0\ket\psi)$ where $\ad_\psi(iH)=-iH\ket\psi$ and $(\cdot)^+$ denotes the Moore--Penrose pseudo-inverse.
The compensating Hamiltonian $H_c$ exactly cancels out the local unitary drift induced by the coupling term $H_0$ at the state $\ket\psi$.
Concrete expressions for $(\ad_\psi^d)^+$ are given in~\cite[App.~A.1]{BipartiteReduced24}.
It plays an important role in lifting solutions~\cite[Prop.~2.13]{BipartiteReduced24} as well as stabilization~\cite[Prop.~3.5]{BipartiteReduced24} as we will see in the following sections.

\section{Optimal Control of the Schmidt Angle}  \label{sec:opt-ctrl-angle}

Let us now turn to the question of optimal control.
Given two quantum states, the challenge is to find a solution connecting them in the least amount of time possible.
Additionally we want to be able to stabilize states with a desired amount of entanglement. 

In this section we consider the cases in which the reduced state space is one-dimensional. 
In these cases the state can be described by a single value, called the Schmidt angle $\chi$. 
We will start with the simplest non-trivial setting of two coupled qubits in Section~\ref{sec:two-qubits} before treating the bosonic case in Section~\ref{sec:bosonic-qubits} and the fermionic case in Section~\ref{sec:fermionic-qu4its}. 
In each section we first derive the exact speed limit of the reduced control system, and then turn to the full control system to compute the corresponding optimal controls and to show how states with a prescribed set of singular values can be stabilized.

\subsection{Two distinguishable qubits} \label{sec:two-qubits}


First we treat the case of two distinguishable qubits with an arbitrary coupling Hamiltonian $H_0$.
Using local unitary control, the state $\ket\psi\in\C^2\otimes\C^2$ can always be brought into diagonal form $\ket\psi=\sigma_1\ket{11}+\sigma_2\ket{22}$, where $\sigma_1,\sigma_2\in\R$ are the singular values of $\ket\psi$.
Due to the normalization of the state $\ket\psi$, it holds that $\sigma_1^2+\sigma_2^2=1$ and hence there exists an angle $\chi$ such that $\sigma_1=\cos(\chi)$ and $\sigma_2=\sin(\chi)$.
Commonly $\chi$ is called the \emph{Schmidt angle}.
Since the singular values are only defined up to sign and order, it suffices to consider states in the region $\chi\in[0,\pi/4]$ (called the Weyl chamber), where $0$ represents product states and $\pi/4$ represents maximally entangled states.
Since the controls $-H_U$ of the reduced control system on this circle are the generators of rotations, i.e.\ $H_U\in\mf{so}(2,\R)$, they are described by their angular velocity $\omega(H_U)=(H_U)_{21}\in\R$. 
One can show that the set of achievable angular velocities is a closed interval symmetric around $0$.
Hence the optimal control task boils down to finding the largest achievable angular velocity of the reduced control system, which we will denote $\omega^\star(H_0)$ or simply $\omega^\star$.
These results yield quantum speed limits on the evolution of the singular values (and hence the amount of entanglement) for the corresponding systems.

\smallskip Using the standard Pauli basis of the $2\times2$ Hermitian matrices
\begin{equation*} \label{eq:paulis}
P_0=\begin{pmatrix}1&0\\0&1\end{pmatrix},\quad
P_x=\begin{pmatrix}0&1\\1&0\end{pmatrix},\quad
P_y=\begin{pmatrix}0&-\iu\\\iu&0\end{pmatrix},\quad
P_z=\begin{pmatrix}1&0\\0&-1\end{pmatrix},
\end{equation*}
and after removing the local part (cf.\ Remark~\ref{rmk:local}) of the coupling Hamiltonian $H_0$, it can, by Lemma~\ref{lemma:normal-form}, be uniquely written as
$$
H_0=\textstyle\sum_{i,j=1}^3 C_{ij}\, P_i\otimes P_j 
$$
where $C=C(H_0)\in\R^{3,3}$ is the coefficient matrix (with the indices $1,2,3$ corresponding to $x,y,z$ in that order).

\begin{lemma} \label{lemma:local-rot-trafo}
Under a local unitary basis transformation $U=V\otimes W$ the coefficient matrix transforms as
\begin{equation*}
C(U^*H_0U) = R_V^\top C(H_0)R_W.
\end{equation*}
\end{lemma}

\begin{proof}
Recall that applying unitary basis transformations to $\mf{su}(2)$ corresponds to three-dimensional rotations.
For a given $U\in\SU(2)$ we write $R_U$ for the resulting rotation in $\SO(3)$. 
More explicitly, for $\vec a\in\R^3$ and $\vec P=(P_x,P_y,P_z)$, we have the relation $\Ad_U(\vec a\cdot\vec P)=(R_U\vec a)\vec P$ where $(R_U)_{i,j=1}^3=\tfrac12\tr(P_i U P_j U^{-1})$, see~\cite[Sec.~3.5, Thm.~I]{Cornwell84v1}.
We see that for $U=V\otimes W$ we obtain
$U^*H_0U
=\sum_{i,j=1}^3V^*P_iV\otimes W^*P_jW
=\sum_{i,j,k,l=1}^3C_{ij}(R_V^\top)_{ki}P_k\otimes (R_W^\top)_{lj}P_l
=\sum_{k,l=1}^3(R_V^\top CR_W)_{kl}P_k\otimes P_l$,
as desired.
\end{proof}
\noindent To avoid confusion, note that here we applied local unitary transformations to the physical state space, whereas in Appendix~\ref{app:ham-decs} we apply basis transformations to the space of Hermitian operators.


\smallskip With this we can derive the exact quantum speed limit $\omega^\star$ for the evolution of the singular values of the state.

\begin{proposition} \label{prop:two-qubits-speed-limit}
Let $\ess_i$ for $i=1,2,3$ denote the singular values of $C$ in non-increasing order (and chosen to be non-negative). 
Then $\omega^\star=\ess_1+\ess_2$. 
\end{proposition}

\begin{proof}
First note that for $U=V\otimes W\in\Uloc(d)$ and coupling Hamiltonian $H_0$, if we define $\tilde H_0=U^*H_0U$, then $H_U=\tilde H_{\id}$.
By Lemma~\ref{lemma:local-rot-trafo} it holds that the coefficient matrix of $\tilde H_0$ in the Pauli basis is $C(\tilde H_0)=C(U^*H_0U) = R_V^\top CR_W$.
Using Proposition~\ref{prop:induced-vf} and Kostant's convexity theorem (cf.~\cite{Kostant73}) for the real singular value decomposition we find that $\omega(H_U)=(\tilde H_{\id})_{12}=(R_V^\top CR_W)_{12}+(R_V^\top CR_W)_{21}\leq\ess_1+\ess_2$. 
Moreover, using the real singular value decomposition it is clear that there exist unitaries $V,W$, and hence rotation matrices $R_V,R_W$, such that
$$
R_V^\top CR_W
=
\begin{pmatrix}
0&\ess_1&0\\
\ess_2&0&0\\
0&0&\ess_3
\end{pmatrix},
$$
and hence the bound is tight.
\end{proof}


\smallskip The next step is to compute time-optimal (and later stabilizing) controls in the full bilinear system.
Let $U=V\otimes W$ be a local unitary achieving the speed limit 
as in the proof of Proposition~\ref{prop:two-qubits-speed-limit}.
Then the path 
$$
\sigma:[0,T]\to S^1, \quad t \mapsto (\cos(\omega^\star t),\sin(\omega^\star t))
$$
is a solution to the reduced control system $\dot\sigma=-H_U\sigma$ with $\sigma(0)=(1,0)$ which reaches the quantum speed-limit of Proposition~\ref{prop:two-qubits-speed-limit}.

Using~\cite[Prop.~2.13]{BipartiteReduced24} one can now determine a lifted solution and the corresponding control functions with possible divergences at the non-regular points $0$ and $\frac\pi4$ corresponding to product states and maximally entangled states respectively. 
Indeed the control Hamiltonian is\footnote{The direct term vanishes since the optimal control unitary is constant in time.}
$$
H_c=E\otimes\id+\id\otimes F
=-\iu \Ad_U \circ (\ad_{\sigma(t)}^d)^{-1}\circ\Pi_\Sigma^\perp(\iu\tilde H_0\ket{\sigma(t)})
$$
where $\Ad_U(\cdot)=U\cdot U^*$ and $\tilde H_0=U^*H_0U$.
Moreover $(\ad_\sigma^d)^{-1}$ and $\Pi_\Sigma^\perp$ are given in~\cite[Lem.~A.7 \& p.~9]{BipartiteReduced24}.
Explicitly, denoting $C'=C(\tilde H_0)$, for $\chi\neq k\tfrac\pi4$ we obtain that%
\footnote{
We use the notation $\sec(\chi)=\frac1{\cos(\chi)}$ for the secant and $\csc(\chi)=\frac1{\sin(\chi)}$ for the cosecant.
}
\begin{align} \label{eq:comp-dist} 
\begin{split}
V^*EV&=-\Big(\begin{smallmatrix}
\frac12 (C'_{zz} + (C'_{xx} - C'_{yy}) \tan(\chi)) & \sec(2\chi) (C'_{xz} - \iu C'_{yz} - (C'_{zx} + \iu C'_{zy}) \sin(2\chi)) \\
\sec(2\chi) (C'_{xz} + \iu C'_{yz} - (C'_{zx} - \iu C'_{zy}) \sin(2 \chi)) & \frac12 (C'_{zz} + (C'_{xx} - C'_{yy}) \cot(\chi))
\end{smallmatrix}\Big) \\
W^*FW&=-\Big(\begin{smallmatrix}
\frac12(C'_{zz} + (C'_{xx} - C'_{yy}) \tan(\chi)) & \sec(2\chi) (C'_{zx} - \iu C'_{zy} - (C'_{xz} + \iu C'_{yz}) \sin(2\chi)) \\
\sec(2\chi) (C'_{zx} + \iu C'_{zy} - (C'_{xz} - \iu C'_{yz}) \sin(2 \chi)) & \frac12 (C'_{zz} + (C'_{xx} - C'_{yy}) \cot(\chi))
\end{smallmatrix}\Big).
\end{split}
\end{align}
In our case $C'_{xy}=\ess_1$, $C'_{yx}=\ess_2$ and $C'_{zz}=\ess_3$ and the remaining matrix entries vanish. 
Hence this becomes $E=F=-\tfrac{\ess_3}2\id$, which is independent of the state $\chi$ and thus also of time.
This proves the following result, where we assume for simplicity that the control Hamiltonians linearly span $\uloc(2,2)$.

\begin{proposition}
For a system composed of two qubits initially in the state $\ket{00}$, the following sequence yields a maximally entangled state and does so in minimal time%
\footnote{Mathematically, even with unbounded controls, the optimal time can only be reached approximately.
Hence one should more accurately speak of infimum time.}:
\begin{enumerate}[1.]
\item Apply the local unitary $U$ (almost) instantaneously.
\item \label{it:ctrl-dist} Apply the constant control Hamiltonian $H_c=-\ess_3\id\otimes\id$ for time $\tfrac{\pi/4}{\ess_1+\ess_2}$.
\item (Optionally) apply a local unitary to choose the desired maximally entangled state.
\item (Optionally) stabilize the final state with an appropriate local compensating Hamiltonian.
\end{enumerate}
The reverse direction is analogous.
\end{proposition}
\noindent Note that the local controls in Step~\ref{it:ctrl-dist}\ only apply a global phase, and so they may be omitted if the global phase is neglected.

\begin{remark}
The fact that the time-optimal solutions in the reduced control system~\eqref{eq:reduced} are given by rotations of constant speed $\pm\omega^\star$ significantly simplifies the problem since integrating the solution is essentially trivial, and when lifting the control to the full control system~\eqref{eq:bilinear} the direct term vanishes. 
In the setting of~\cite{QubitReduced24} for example this simplification does not occur.
\end{remark}


\smallskip Another important task is that of stabilizing a state such as the maximally entangled state obtained in the previous section.
More precisely, here we want to stabilize a state with a certain set of singular values.
One can show abstractly, cf.~\cite[Prop.~18]{BipartiteReduced24}, that for regular states this is always possible with a fixed control Hamiltonian.
It is easy to see that by choosing a local unitary $U=V\otimes W$ such that $C'=C(U^*H_0U)$ is diagonal, it holds that $H_U\sigma=0$ for any $\sigma\in S^1$.
Thus, from~\eqref{eq:comp-dist} we find that a corresponding local stabilizing Hamiltonian is given by
\begin{align*}
V^*EV=W^*FW
=-\tfrac12(C'_{zz}+(C'_{xx}-C'_{yy})\csc(2\chi))\id + \tfrac12(C'_{xx}-C'_{yy})\cot(2\chi) P_z.
\end{align*}
Indeed, this also works for a maximally entangled state $\chi=\pi/4$, but diverges when approaching product states $\chi\to0$, even if we ignore the global phase.
In that case a different control Hamiltonian does the job.
As $\chi\to0$ the only terms in~\eqref{eq:comp-dist} which might blow up are those containing $\cot(\chi)$. Thus, if $U$ is chosen such that $C'_{xx}=C'_{yy}$, which can always be achieved, this does not happen.
The resulting compensating Hamiltonian $H_c=E\otimes\id+\id\otimes F$ is then defined by
\begin{align*}
V^*EV&=-\tfrac{C'_{zz}}2\id - \sec(2\chi)(C'_{yz}+C'_{zy}\sin(2\chi))P_y, \\
W^*FW&=-\tfrac{C'_{zz}}2\id - \sec(2\chi)(C'_{zy}+C'_{yz}\sin(2\chi))P_y.
\end{align*}
The examples above show that there might exist many local unitaries $U=V\otimes W$ such that $H_U=0$, and each choice yields its own compensating Hamiltonian.
In particular by choosing the right $U$ we were able to prevent the controls from blowing up at product states and maximally entangled states.

\subsection{Two bosonic qubits} \label{sec:bosonic-qubits}


In addition to the systems composed of two distinguishable subsystems considered above, one may also consider indistinguishable subsystems.
Such systems are characterized by the fact that swapping the two subsystems changes the state only up to a global phase.
If this phase is $+1$ the system is called bosonic, and if it is $-1$ the system is called fermionic.
The theory goes through with only minor adaptations in this indistinguishable setting, see~\cite{BipartiteReduced24} for the details.
In particular the coupling and control Hamiltonians have to be symmetric under swapping as well.
We denote the symmetric local unitary Lie group and algebra by $\Uloc^s(d)$ and $\uloc^s(d)$ respectively.

Let us now consider the case of two bosonic qubits. 
Again the coupling Hamiltonian can be expressed in the Pauli basis as
$$
H_0=\textstyle\sum_{i,j=1}^3 C_{ij}\, P_i\otimes P_j, \quad C\in\R^{3,3}.
$$
This time the coefficient matrix $C$ is symmetric. 
Transforming the coupling Hamiltonian using a local unitary $U=V\otimes V\in\Uloc^s(d)$, it follows from Lemma~\ref{lemma:local-rot-trafo} that $C(U^*H_0U)=R_V^\top C(H_0)R_V$.

Here and in the next section we will make use of the following simple lemma about symmetric and Hermitian matrices of size $2\times2$:

\begin{lemma} \label{lemma:2by2} 
Consider a Hermitian matrix $H\in\iu\mf u(2)$ and let $\ell_1\geq\ell_2$ denote its eigenvalues. 
Then it holds that $|H_{12}|\leq\frac{\ell_1-\ell_2}{2}$ and there is a unitary $U\in\SU(2)$ such that $|(U^*HU)_{12}|=\frac{\ell_1-\ell_2}{2}$.
The analogous statement for a real symmetric $2\times2$ matrix and orthogonal conjugation also holds.
\end{lemma}

\begin{proof}
Due to the unitary invariance of the Frobenius norm it holds that
$(H_{11})^2 + 2|H_{12}|^2 + (H_{22})^2 = \ell_1^2 + \ell_2^2$.
Since $H_{11}+H_{22}=\ell_1+\ell_2$ there is some $x\in\R$ such that $H_{11}=\frac{\ell_1+\ell_2}{2}+x$ and $H_{22}=\frac{\ell_1+\ell_2}{2}-x$.
Together this gives
$$
2|H_{12}|^2 
= \ell_1^2 + \ell_2^2 - (\tfrac{\ell_1+\ell_2}{2}+x)^2 - (\tfrac{\ell_1+\ell_2}{2}-x)^2
= \tfrac{(\ell_1-\ell_2)^2}{2}-2x^2.
$$
This proves the desired bound on $|H_{12}|$. It remains to show that there is a unitary $U$ such that the diagonal elements of $U^*HU$ coincide. For this we may assume that $H$ is diagonal. Then we compute
$$
\frac{1}{2}
\begin{pmatrix} 1 & 1 \\ 1 & -1 \end{pmatrix}
\begin{pmatrix} \ell_1 & 0 \\ 0 & \ell_2 \end{pmatrix}
\begin{pmatrix} 1 & 1 \\ 1 & -1 \end{pmatrix}
=
\frac{1}{2}\begin{pmatrix} \ell_1+\ell_2 & \ell_1-\ell_2 \\ \ell_1-\ell_2 & \ell_1+\ell_2 \end{pmatrix},
$$
and this concludes the proof for the Hermitian case.
The real symmetric case is almost identical.
\end{proof}


This time the exact speed limit is given in terms of the eigenvalues of the coefficient matrix $C$:

\begin{proposition}
Let $\ell_i(C)$ for $i=1,2,3$ denote the eigenvalues of $C$ in non-increasing order.
Then $\omega^\star=\ell_1(C)-\ell_3(C)$.
\end{proposition}

\begin{proof} 
The idea is similar to the proof of Proposition~\ref{prop:two-qubits-speed-limit}. 
If $U=V\otimes V$, then by Lemma~\ref{lemma:local-rot-trafo} it holds that $C(U^*H_0U)=R_V^\top CR_V$.
Again it holds that $\omega(H_U)=(\tilde H_{\id})_{12}=C(U^*H_0U)_{12}+C(U^*H_0U)_{21}$.
Due to Kostant's Convexity Theorem~\cite{Kostant73}, or more precisely the Schur--Horn Theorem~\cite{Schur23,Horn54}, it holds that $|C(U^*H_0U)_{11}-C(U^*H_0U)_{22}|\leq\ell_1(C)-\ell_3(C)$,
and hence if we denote by $\tilde C(U^*H_0U)$ the upper left $2\times2$ block in $C(U^*H_0U)$, 
it holds that $|\ell_1(\tilde C(U^*H_0U))-\ell_2(\tilde C(U^*H_0U))|\leq\ell_1(C)-\ell_3(C)$.
By Lemma~\ref{lemma:2by2}, $C(U^*H_0U)_{12}+C(U^*H_0U)_{21}\leq\ell_1(C)-\ell_3(C)$ and by the proof of the same lemma this bound is tight.
Explicitly, the bound is achieved by choosing $U=V\otimes V$ such that
\begin{align*}
C(U^*H_0U)=R_V^\top CR_V=
\begin{pmatrix} 
\tfrac{\ell_1+\ell_3}2 & \tfrac{\ell_1-\ell_3}2 & 0 \\ 
\tfrac{\ell_1-\ell_3}2 & \tfrac{\ell_1+\ell_3}2 & 0 \\
0 & 0 & \ell_2
\end{pmatrix}.
\end{align*}
This concludes the proof.
\end{proof}


Denoting $C'=C(U^*H_0U)=R_V^\top CR_V$, the compensating Hamiltonian is given by $H_c=E\otimes\id+\id\otimes E$ where
$$
V^*EV=-\Big(\begin{smallmatrix}
\tfrac12 (C'_{zz} + (C'_{xx} - C'_{yy}) \tan(\chi)) & \sec(2\chi) (C'_{xz} - \iu C'_{yz} - (C'_{xz} + \iu C'_{yz}) \sin(2\chi)) \\
\sec(2\chi) (C'_{xz} + \iu C'_{yz} - (C'_{xz} - \iu C'_{yz}) \sin(2 \chi)) & \tfrac12 (C'_{zz} + (C'_{xx} - C'_{yy}) \cot(\chi))
\end{smallmatrix}\Big),
$$
where we used the results of~\cite[App.~A]{BipartiteReduced24}.

With these results, the derivation of time-optimal and stabilizing controls is straightforward and analogous to the previous section.


\subsection{Two fermionic four-level systems} \label{sec:fermionic-qu4its}

In the fermionic case the space of singular values has dimension $\floor{\tfrac{d}2}$.
Hence, if we consider two coupled four-level systems ($d=4$), the reduced state space is again a circle. 
Due to Lemma~\ref{lemma:normal-form-sym} we can always write the coupling Hamiltonian in the diagonal form. 
Here we will focus on the rank one case, i.e., $H=A\otimes A$ for some Hamiltonian $A\in\iu\mf u(4)$.

\begin{lemma} \label{lemma:fermi-reduced} 
In the fermionic $d=4$ case the reduced control system can equivalently be formulated for the Schmidt angle $\chi$ as
$\dot\chi = \omega_U^a$, where $\omega_U^a = \Im((U^*AU)_{13}(U^*AU)_{24}-(U^*AU)_{14}(U^*AU)_{23})$.
\end{lemma}

\begin{proof}
This follows immediately from Proposition~\ref{prop:induced-vf}.
\end{proof}


We begin by deriving upper and lower speed limits for the Schmidt angle $\chi$.

\begin{proposition} \label{prop:fermion-speed-limit}
Let $\ell_i$ denote the eigenvalues of $A$ in non-increasing order.
Then it holds that 
$$
\tfrac14(\ell_1-\ell_3)(\ell_2-\ell_4)\leq\omega^*\leq\tfrac1{16}(\ell_1+\ell_2-\ell_3-\ell_4)^2,
$$
and the two bounds coincide when $\ell_1+\ell_4=\ell_2+\ell_3$.
\end{proposition}

\begin{proof}
First we show that we can assume certain elements of $A$ to vanish.
Let $A_{(12)}$ denote the $2\times2$ block of $A$ in the upper right corner.
Then, by Lemma~\ref{lemma:fermi-reduced} it holds that $\omega_U^a=\Im(\det((U^*AU)_{(12)}))$. 
Using a block diagonal unitary change of basis and the complex singular value decomposition we may assume that $A_{(12)}$ is diagonal and additionally that $\omega_{\id}^a=|A_{13}A_{24}|$.
Now let $A_{[ij]}$ denote the $2\times2$ submatrices of $A$ obtained by deleting all but the $i$-th and $j$-th row and column.
Let $a\geq c$ denote the eigenvalues of $A_{[13]}$ and $b\geq d$ those of $A_{[24]}$.
Then by Lemma~\ref{lemma:2by2} it holds that $\omega_U^a=\tfrac14(a-c)(b-d)$, and there exists a unitary transformation $U$ such that the diagonal of $U^*AU$ is $(a,b,c,d)$.
Hence we need to solve the optimization problem
$$
\max \tfrac14(a-c)(b-d) \text{ subject to } (a,b,c,d)\preceq(\ell_1,\ell_2,\ell_3,\ell_4).
$$
This can be done using a greedy optimization approach.
First we show that the maximum is achieved when $a-c=b-d$.
Indeed, if $a-c>b-d$ we can smoothly move in the direction $(-1,1,0,0)$, which preserves majorization and increases the objective value, until equality is achieved. Similarly, if $a-c<b-d$ we move in the direction $(0,0,-1,1)$.
Finally, by moving in the direction $(-1,1,-1,1)$, which does not affect the objective value, we may assume that additionally $a=b$ and $c=d$.
Under these additional constraints the maximum is easily seen to be $\tfrac1{16}(\ell_1+\ell_2-\ell_3-\ell_4)^2$, as desired.

By Lemma~\ref{lemma:2by2} the lower bound can be achieved as follows. 
First diagonalize $A$ to obtain the form $\diag(\ell_1,\ell_2,\ell_3,\ell_4)$.
Then, using a unitary mixing the levels $1$ and $3$, as well as $2$ and $4$, as in the proof of Lemma~\ref{lemma:2by2}, we obtain the form
$$
U^*AU= 
\frac12
\begin{pmatrix}
\ell_1+\ell_3&0&\ell_1-\ell_3&0\\
0&\ell_2+\ell_4&0&\ell_2-\ell_4\\
\ell_1-\ell_3&0&\ell_1+\ell_3&0\\
0&\ell_2-\ell_4&0&\ell_2+\ell_4
\end{pmatrix},
$$
which achieves the claimed lower bound.
\end{proof}

\begin{remark}
To see that the upper bound of Proposition~\ref{prop:fermion-speed-limit} is not tight in general, consider $A$ of rank one, e.g., $A=\diag(1,0,0,0)$. Then it is easy to verify that $\omega_U^a\equiv0$ for all $U\in\Uloc^s(4,4)$.
\end{remark}


\smallskip In a basis as described by Proposition~\ref{prop:fermion-speed-limit}, the local compensating Hamiltonian takes the form
$E\otimes\id+\id\otimes E$ where $V^*EV=-\tilde E-\tfrac18(\ell_1+\ell_3)(\ell_2+\ell_4)\id$ and
$$
\tilde E=\begin{pmatrix}
\tfrac{(\ell_1-\ell_3)(\ell_2-\ell_4)\tan(\chi)}{8} & 0 & \tfrac{(\ell_1-\ell_3)(\ell_2+\ell_4)}{4} & 0 \\
0 & \tfrac{(\ell_1-\ell_3)(\ell_2-\ell_4)\tan(\chi)}{8} & 0 & \tfrac{(\ell_1+\ell_3)(\ell_2-\ell_4)}{4} \\
\tfrac{(\ell_1-\ell_3)(\ell_2+\ell_4)}{4} & 0 & \tfrac{(\ell_1-\ell_3)(\ell_2-\ell_4)\cot(\chi)}{8} & 0 \\
0 & \tfrac{(\ell_1+\ell_3)(\ell_2-\ell_4)}{4} & 0 & \tfrac{(\ell_1-\ell_3)(\ell_2-\ell_4)\cot(\chi)}{8}
\end{pmatrix}.
$$
This Hamiltonian blows up near product states as $\chi\to k\pi/2$ but not near maximally entangled states.


\smallskip By choosing a local basis which diagonalizes $A$, we obtain the compensating local Hamiltonian
$V^*EV=-\diag(\tfrac{\ell_1\ell_2}2,\tfrac{\ell_1\ell_2}2,\tfrac{\ell_3\ell_4}2,\tfrac{\ell_3\ell_4}2)$,
which is independent of $\chi$.

\section{Optimal Control of Qutrits} \label{sec:opt-ctrl-qutrits}

In this section we consider a higher dimensional system, namely one composed of two distinguishable three-level systems (qutrits).
In this case there are three singular values and hence the state space of the reduced control system is the usual two-dimensional sphere $S^2$.
Compared to the previous section, it is now not at all obvious which path to take between two points on the Schmidt sphere in order to minimize the time.
To determine such an optimal path we will make use of the Pontryagin Maximum Principle (PMP)~\cite{Boscain21}.

In Section~\ref{sec:reduce} we define the reduced control system by characterizing the set of generators $\mf H$.
In Section~\ref{sec:pmp} we use the PMP to find time-optimal solutions to the reduced control system.
Finally we (approximately) lift these solutions to the original control system and derive corresponding control functions in Section~\ref{sec:lift}.

\subsection{Reducing the problem} \label{sec:reduce}

In this section we assume that the two distinguishable subsystems are qutrits, i.e.\ $d_1=d_2=3$, and that the coupling Hamiltonian $H_0=A\otimes B$ has rank one.
Moreover we assume that $A$ and $B$ have equidistant eigenvalues.%
\footnote{We say that $A$ has equidistant eigenvalues if $\ell_1(A)-\ell_2(A)=\ell_2(A)-\ell_3(A)$.}
We start with a simple lemma which bounds the size of the off-diagonal elements of $A$ and $B$ in terms of their respective eigenvalues.

\begin{lemma} \label{lemma:norm-bound}
Let $A\in\iu\mf{u}(3)$ with equidistant eigenvalues be given and let $a=(A_{32},A_{13},A_{21})\in\C^3$. 
Then it holds that $\|a\|_2\leq \tfrac12(\ell_{1}-\ell_{3})$ where the $\ell_{i}$ denote the eigenvalues of $A$ in non-increasing order.
\end{lemma}

\begin{proof}
Since the expression $\|a\|_2\leq \tfrac12(\ell_{1}-\ell_{3})$ is invariant under addition of a multiple of the identity to $A$, we may assume that $A$ is traceless, and hence the eigenvalues are $-\ell\leq0\leq\ell$. 
Now it holds that $2\ell^2=\|A\|_2^2\geq 2\|a\|_2^2$ and hence $\|a\|_2\leq\ell$ as desired.
\end{proof}

The first step in solving the optimal control problem is to understand the reduced control system.
We already know that the reduced control system is defined on the sphere $S^2$ and that the controls are given by the set $\mf H\subset\mf{so}(3)$ consisting of rotation generators.
The goal of this section is to understand the precise shape of this set (or at least its convex hull).
It is convenient to represent generators in $\mf{so}(3)$ using vectors in $\R^3$.
Indeed, for every $-H_U\in\mf H$ there is a unique vector $\omega_U\in\R^3$ such that $-H_U\sigma=\omega_U\times\sigma$.
Recall that the $1$-norm of such a vector is given by $\|\omega\|_1:=|\omega_1|+|\omega_2|+|\omega_3|.$

\begin{proposition} \label{prop:octahedral-bound}
Assume that $A$ and $B$ have equidistant eigenvalues. Then, for all $U\in\Uloc(3,3)$ it holds that
$$
\|\omega_U\|_1\leq \omega^\star(H_0) := \frac{(\ell_{1}(A)-\ell_{3}(A))(\ell_{1}(B)-\ell_{3}(B))}4.
$$
\end{proposition}

\begin{proof}
Using Proposition~\ref{prop:induced-vf}, the Cauchy--Schwarz inequality and Lemma~\ref{lemma:norm-bound} we compute
$$\textstyle \|\omega_U\|_1
=\sum_{i\in\{x,y,z\}} |\Im(\tilde a_i \tilde b_i)|
\leq
\sum_{i\in\{x,y,z\}} |\tilde a_i||\tilde b_i|
\leq
\|\tilde a\|_2\|\tilde b\|_2
\leq
\omega^\star,
$$
where, denoting $U=V\otimes W$, we set $\tilde a=((V^*AV)_{32},(V^*AV)_{13},(V^*AV)_{21})$ and similar for $\tilde B$.
\end{proof}
\noindent Geometrically this has a nice interpretation.
Let $O_3\subset\R^3$ denote the regular octahedron, i.e., the convex hull $O_3=\conv((\pm1,0,0),(0,\pm1,0),(0,0,\pm1))$, see Figure~\ref{fig:octahedron}.
Note also that this bound is stronger than the general bound obtained in~\cite[Sec.~3.2]{BipartiteReduced24}.

\begin{corollary}
The convex hull of the set of induced vector fields considered in $\R^3\cong\mf{so}(3)$ is a regular octahedron:
$$
\conv(\omega_U:U\in\Uloc(3,3)) = \omega^\star(H_0)\, O_3.
$$
Thus, in the relaxed control system, two reduced states $\sigma,\tau\in S^2$ can always be joined in time $T\in\omega^\star(H_0) \arccos(\sigma\cdot\tau)[1/\sqrt3,1]$.
\end{corollary}

\begin{proof}
The inclusion $\subseteq$ follows immediately from Proposition~\ref{prop:octahedral-bound}. To get equality one just has to obtain the vertices of the octahedron, which can be done in a manner similar to the proof of Lemma~\ref{lemma:2by2}.
The bounds on $T$ follow immediately from the fact that the spheres of radius $1/\sqrt3$ and $1$ are respectively the inscribed and circumscribed spheres of the regular octahedron. See also~\cite[Sec.~3.2]{BipartiteReduced24}.
\end{proof}
\noindent Recall from Remark~\ref{rmk:different-reduced} that one can define a relaxed control system via $\dot\sigma(t)\in\conv(\mf{H}\sigma)$, which is still approximately equivalent to~\eqref{eq:reduced}. For this reason we will work with the relaxed system when we derive time-optimal solutions in the following section.

\subsection{Solving the reduced problem} \label{sec:pmp}

By rescaling we may assume without loss of generality that $\omega^\star=1$ and hence the reduced control system becomes 
\begin{align*}
\dot\sigma(t) = u(t)\times\sigma(t), 
\quad\sigma(0)=\sigma_0\in S^2
\end{align*}
where the control function $u:[0,T]\to\R^3$ is measurable and satisfies 
\begin{equation} \label{eq:constraint}
\|u(t)\|_1=|u_x(t)|+|u_y(t)|+|u_z(t)|\leq 1 
\end{equation}
for almost all $t\in[0,T]$.

\smallskip The time-optimal control problem can be solved using the Pontryagin Maximum Principle (PMP) \cite{Pont86,Agrachev04,Sachkov22}.
An introduction to the PMP in the context of quantum control theory is given in~\cite{Boscain21}. 
The PMP is a first-order necessary condition satisfied by optimal trajectories. 
We introduce an \emph{adjoint state} $p\in\R^3$ and define the \emph{pseudo-Hamiltonian} $H_{\sf p}(u,\sigma,p)=p\cdot(u\times\sigma)$.
The dynamics of $\sigma$ and $p$ follow the Hamilton--Jacobi equations, in particular $\dot{p}=-\partial H_{\sf p}/\partial\sigma=u\times p$.
For this kind of control system on the sphere, we can define the variable%
\footnote{Strictly speaking $p$ is a cotangent vector $p\in T^*_\sigma S^2$. When considered as a vector in $\R^3$, it is therefore restricted to be orthogonal to $\sigma$ at all times. Thus the relation between $p$ and $l$ is bijective.}
$l=\sigma\times p$ (note that this implies that $l$ is orthogonal to $\sigma$) that allows us to express the pseudo-Hamiltonian in the convenient form:
\begin{equation}
\tilde H_{\sf p}(u,l)=u\cdot l.\label{eq:pseudoH}
\end{equation}
We can show using the Hamilton--Jacobi equations that the dynamics of $l$ follow:
\begin{equation} \label{eq:dynamicl}
\dot l
= \dot\sigma\times p+\sigma\times\dot p
=(u\times\sigma)\times p+(p\times u)\times\sigma
=u\times l.
\end{equation}
The PMP states (see~\cite[Thm.~5]{Boscain21}) that the time-optimal control has to maximize this pseudo-Hamiltonian under the constraint~\eqref{eq:constraint}. 
For this we determine which values of $u$ (under the given constraint) maximize $\tilde H_{\sf p}(u(t),l(t))$ for given $l$: 

\begin{lemma} \label{lemma:opt-u}
Given $l\in\R^3$, the vertex $\pm e_i$ of $O_3$ maximizes $\tilde H_{\sf p}(u,l)=u\cdot l$ if and only if $\pm l_i=\max(|l_x|,|l_y|,|l_z|)$. The set of all maximizers is then simply the convex hull of such vertices, and defines a face of the octahedron.
\end{lemma}

\begin{proof}
Due to the normalization of $u$ it is clear that $u\cdot l\leq\max(|l_x|,|l_y|,|l_z|)$. 
Moreover, since the maximization is linear and takes place on the octahedron, which is a convex polytope, the maximum is achieved exactly on a face (which may be a vertex, an edge, or a facet) of the octahedron, and hence defined by a subset of vertices. 
It is easy to see that the vertices which maximize $u\cdot l$ are exactly the ones given in the statement, and hence the set of maximizers is the convex hull of these vertices. 
\end{proof}

\begin{remark} \label{rmk:uopt}
This result can be visualized intuitively by considering the polar dual of the octahedron $O_3$, which is nothing but the cube with vertices $(\pm1,\pm1,\pm1)$ denoted $C_3$.
If $l$ lies on a certain face of the cube, then there is a unique corresponding face of the octahedron containing all $u$ which maximize $H_{\sf p}=u\cdot l$.
Additionally it turns out that only the barycenters of the faces of the octahedron yield relevant values of $u$. 
This will become clear later when we describe all possible evolutions of $l$.
The duality and the barycenters are shown in Figure~\ref{fig:octahedron}.
\end{remark}

\begin{figure}[th]
\includegraphics[scale=0.3]{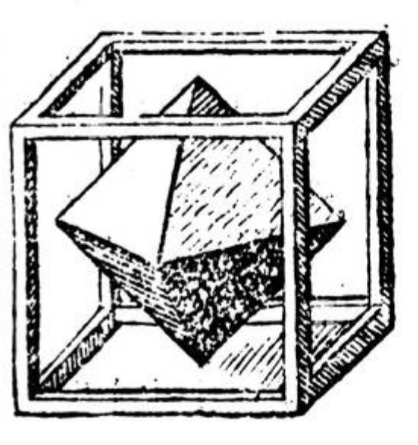}
\hspace{1cm}
\includegraphics[scale=0.3]{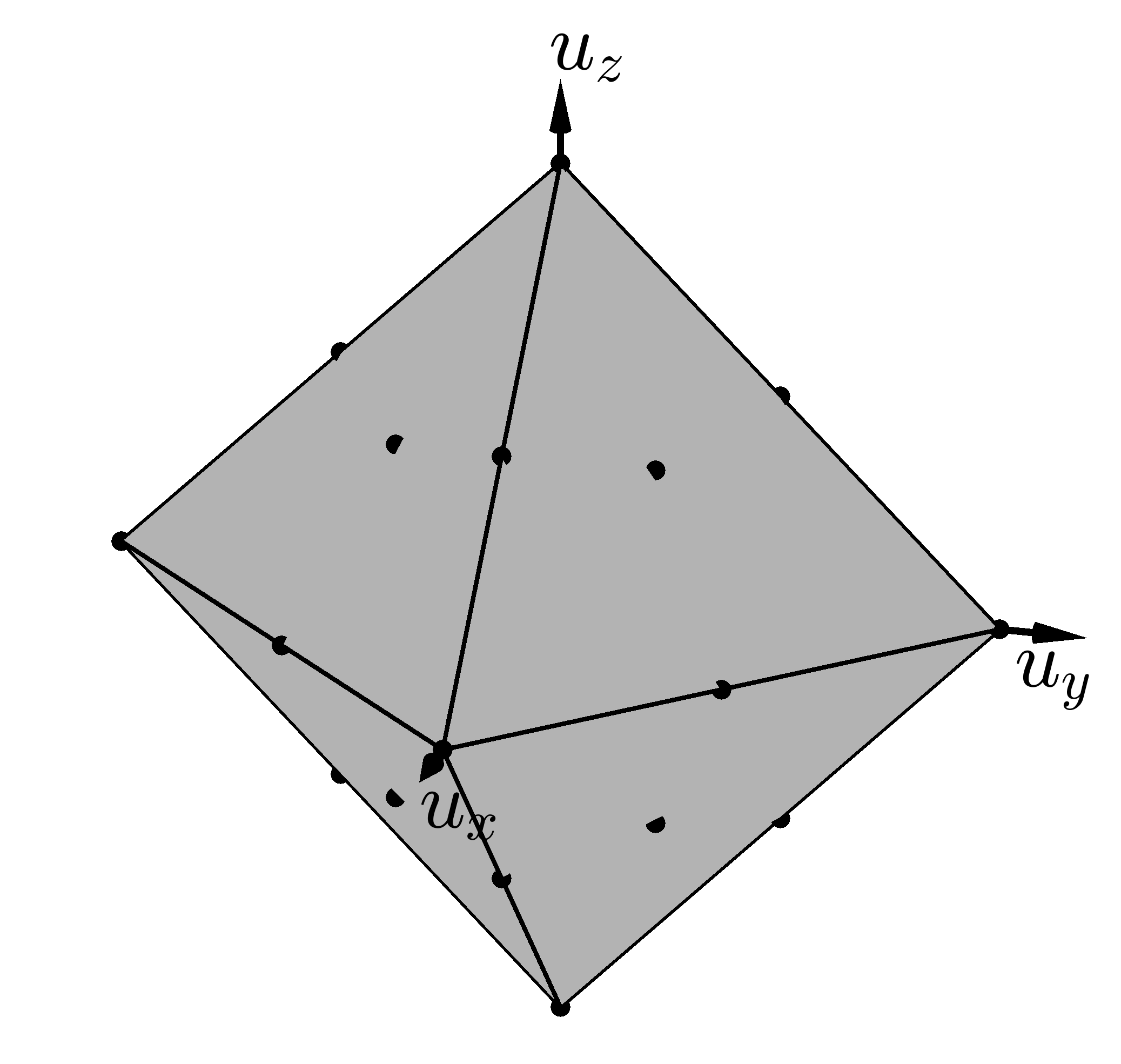}
\caption{
Left: The octahedron $O_3$ and its polar dual, the cube $C_3$. Reproduced from~\cite{Kepler}.
Note that the $d$-dimensional faces of $O_3$ correspond one-to-one with $(2-d)$-dimensional faces of $C_3$, that is, vertices correspond to facets and edges to edges.
Right: The control $u$ can take values at the black points which are the barycenters of the faces of the octahedron $O_3$, cf.\ Remark~\ref{rmk:uopt}.}
\label{fig:octahedron}
\end{figure}

The dynamics of $l$ are given by~\eqref{eq:dynamicl} and constrained by Lemma~\ref{lemma:opt-u}.
Since the optimal value of $u$ is not always uniquely defined, this yields a differential inclusion instead of a differential equation.
Thus the solution is in general not uniquely determined by the initial condition $l(0)$. 
However we will see that unique solutions can still be obtained for the optimal control problem.

To understand the evolution of $l$, note that it has two constants of motion. 
The first one is due to the fact that $l$ moves on a sphere (recall that $\dot l=u\times l$), and the second one, referred to as the \emph{Pontryagin Hamiltonian}, is obtained by substituting any optimal $u$ in the pseudo-Hamiltonian~\eqref{eq:pseudoH}. 
They are given by:
\begin{equation} \label{eq:HL2FirstInt}
L^2=|l_x|^2+|l_y|^2+|l_z|^2, \quad
\mathcal{H}=\max\{|l_x|,|l_y|,|l_z|\}.
\end{equation}
The first equation corresponds to a sphere of radius $L$ which can be set to $1$ without loss of generality, and the second to a cube of side length $2\mathcal{H}$. 
Any solution of the system must remain on to the intersection of these two surfaces, illustrated in Fig.~\ref{fig:spherecube}.

\begin{figure}[th]
\includegraphics[scale=0.25]{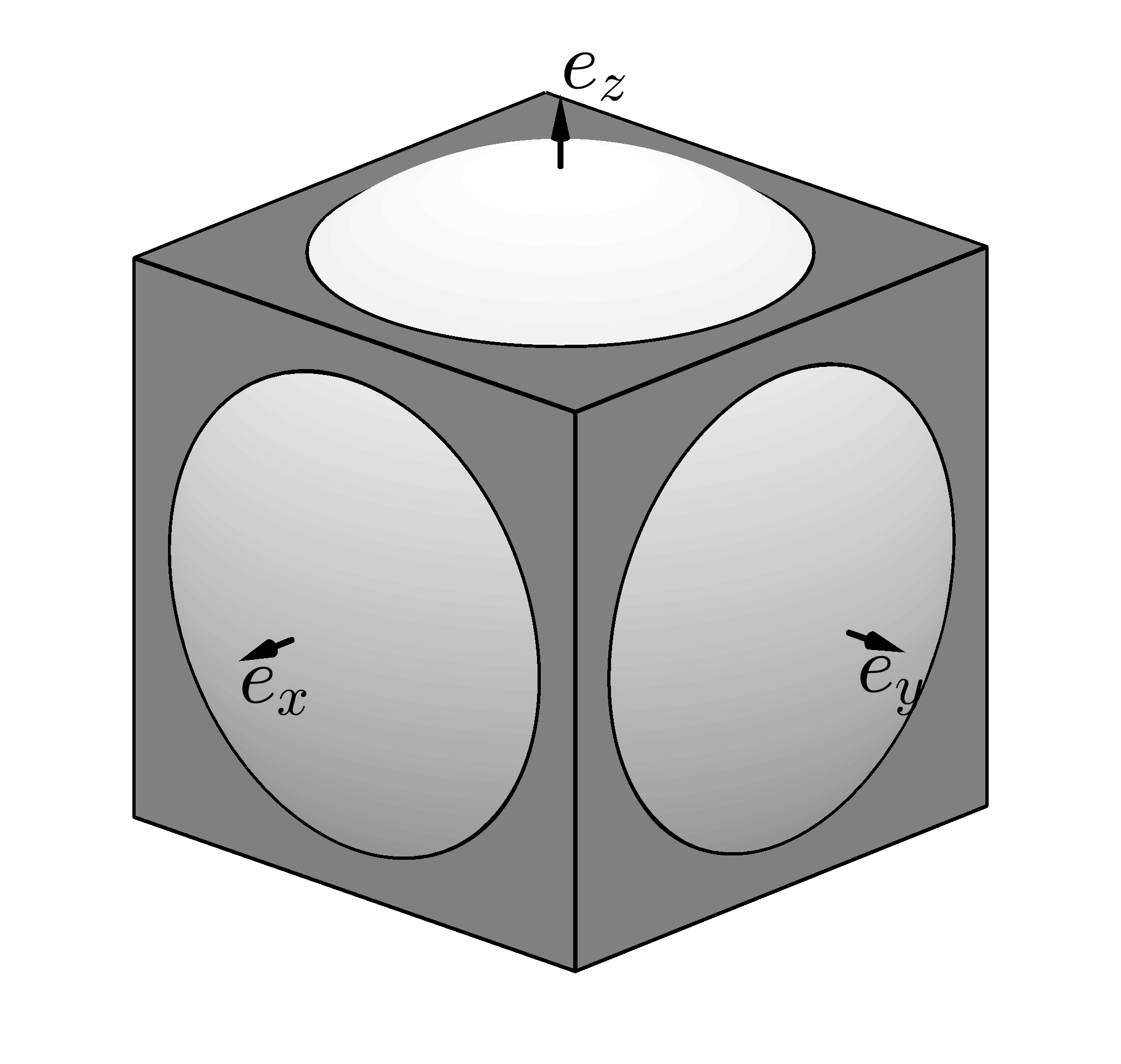}
\includegraphics[scale=0.25]{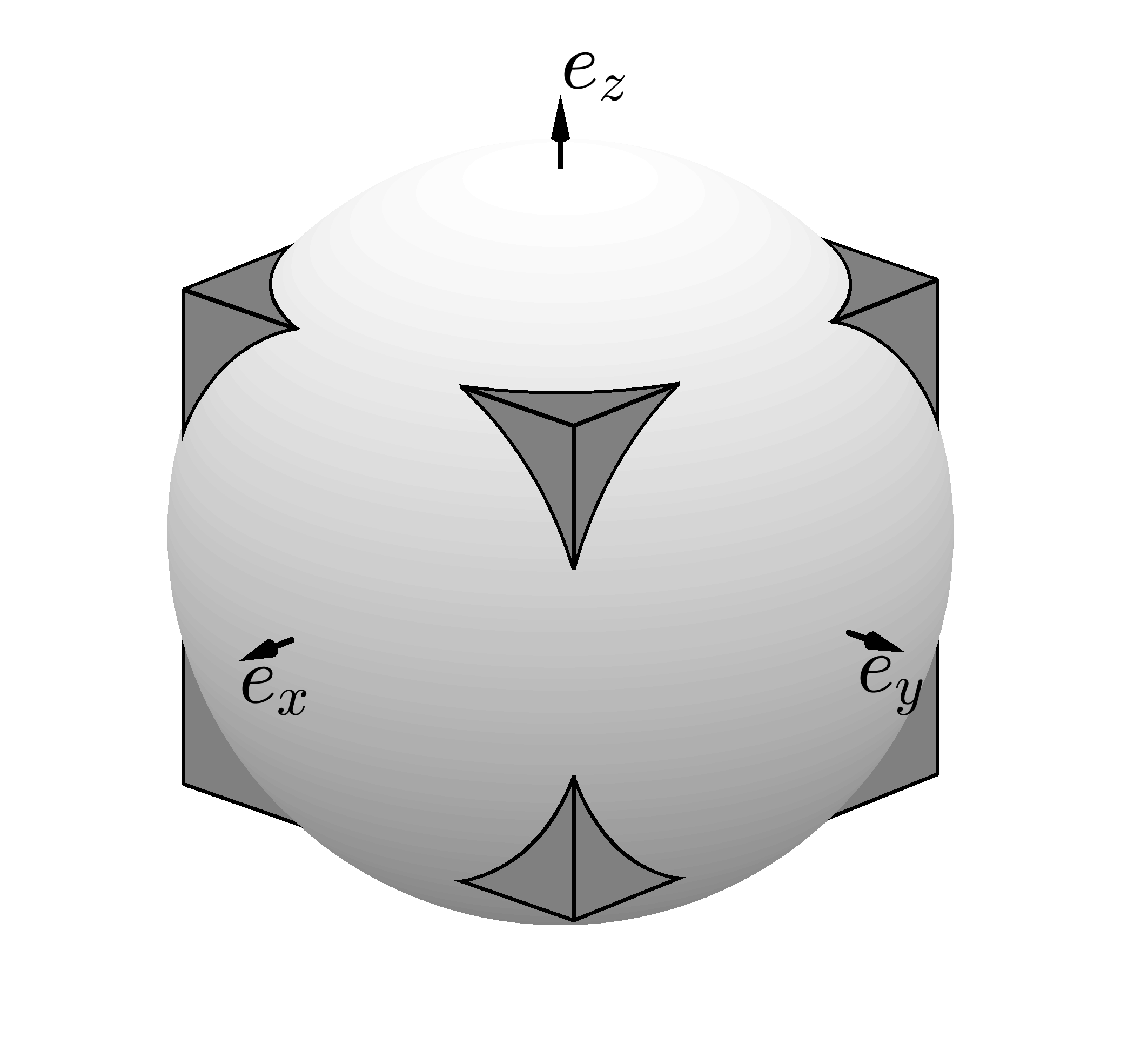}
\includegraphics[scale=0.25]{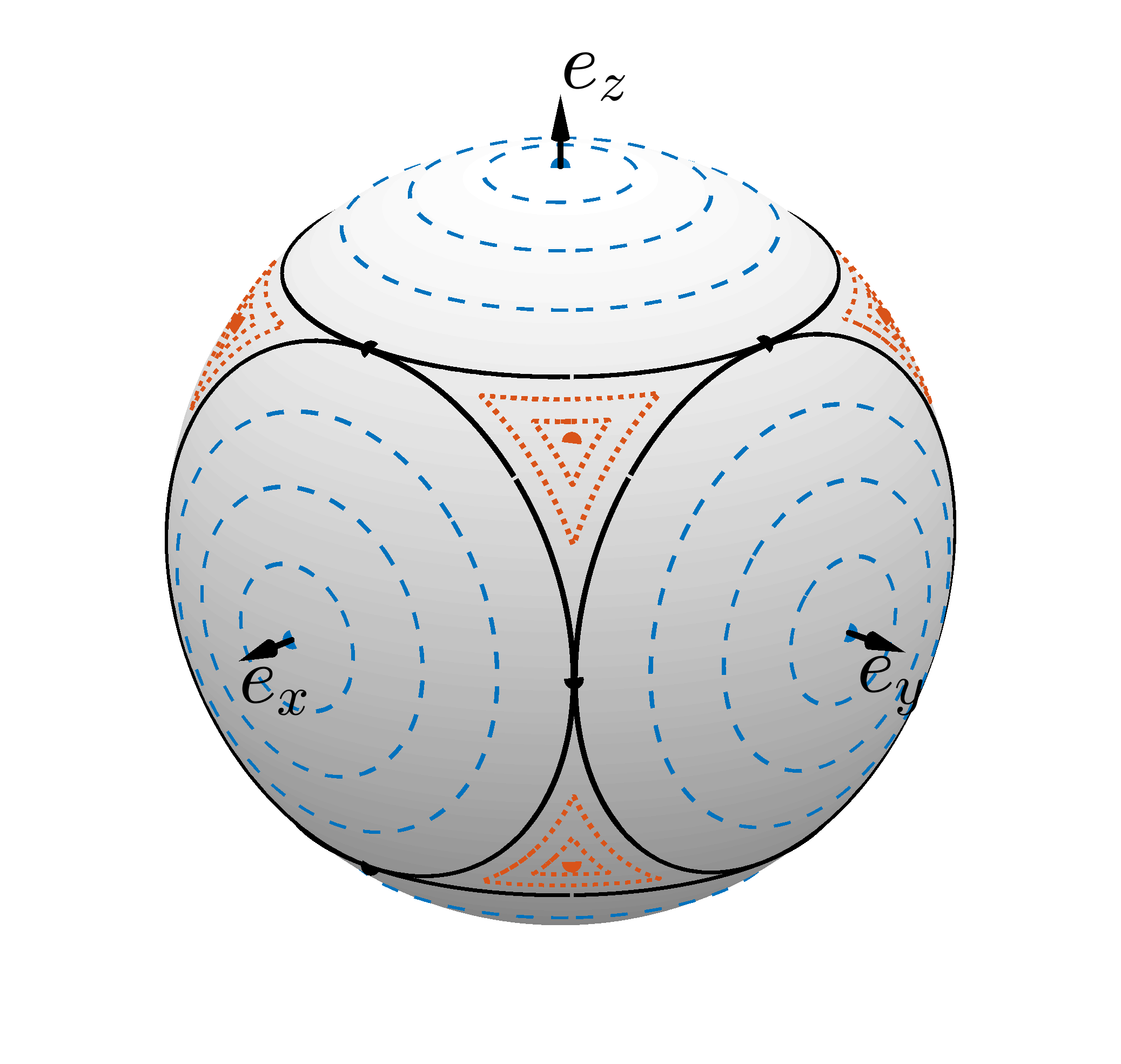}
\caption{Geometric representation of the two first integrals given in Eq.~\eqref{eq:HL2FirstInt}. 
The solution $l(t)$ lies on the intersection of these two surfaces. We obtain different families shown on the rightmost panel, namely the constant case (blue dashed lines) for $\mathcal{H}\in(\tfrac1{\sqrt2},1]$, the switching case (dotted red lines) for $\mathcal{H}\in[\tfrac1{\sqrt3},\tfrac1{\sqrt2})$, and the separatrix (black line) for $\mathcal{H}=\tfrac1{\sqrt2}$.}
\label{fig:spherecube}
\end{figure} 

The general solution $l(t)$ can be decomposed into two main families, namely the \emph{switching case} and the \emph{constant case}, depending on $\mathcal{H}$ (which depends on $l(0)$). 
In the constant case, one of the components of $l$ is always dominant. 
These curves imply constant controls. 
In the switching case, $l$ is made by concatenation of three circular arcs. 
In the positive octant, the control jumps between $u_x\rightarrow u_z\rightarrow u_y\rightarrow u_x\rightarrow\cdots$, with a duration $\Delta t=\pi/2-2\arccos(\mathcal{H}/\sqrt{1-\mathcal{H}^2})$ between two switches.
For $\mc H=\tfrac1{\sqrt3}$ the switching solution degenerates into a solution with constant controls of the form $u=(\pm\tfrac13,\pm\tfrac13,\pm\tfrac13)$, corresponding to the barycenters of the facets of the octahedron of Fig.~\ref{fig:octahedron}.
A special case is the black curve separating these two families. 
If $ l(0)$ starts somewhere on this curve, it follows it for a while until it reaches one of the unstable equilibrium points (in black).
It can stay on this point for a certain amount of time and then follow any of the trajectories connected to this point. These dynamics correspond to a control that is originally such that, for example, $u=(1,0,0)$ during a certain time and then switches to, e.g, $u=(\tfrac12,\tfrac12,0)$ (if the unstable point is in the $xy$-plane) and stays for a certain time. It can thus continue with $u=(1,0,0)$, or switch to $u=(0,1,0)$. 
The time it stays on the unstable equilibrium depends on the trajectory $\sigma(t)$ in the state space that one wants to achieve, and in particular on the desired final state.

So far we have considered the reduced control system on the entire sphere.
However, since the singular values of the quantum state are only defined up to order and sign, there is an additional symmetry, and we may focus on the part of the sphere where $\sigma_z\geq\sigma_y\geq\sigma_x\geq0$. 
This is called the \emph{Weyl chamber} and illustrated in Figure~\ref{fig:weyl}.
Indeed, for any solution of the reduced control system one can consider the corresponding path in the Weyl chamber obtained by taking the absolute value of the singular values and ordering them appropriately, and this is also guaranteed to be a solution, see~\cite[Prop.~A.4]{Reduced23}.

\begin{figure}
\centering
\includegraphics[width=0.45\textwidth,trim=2cm 5cm 2cm 4cm,clip]{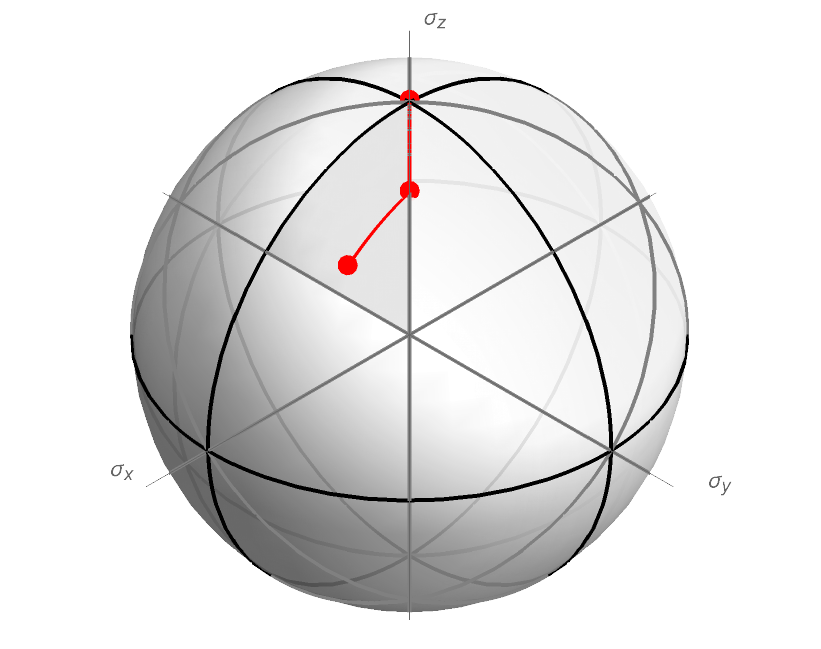}
\caption{The singular values of the quantum state are only defined up to order and sign. On the Schmidt sphere, coordinate permutations and sign flips act via reflections. 
These reflections divide the sphere into small triangular regions, which can be mapped to each other uniquely and are hence equivalent. 
We choose one of these triangles (namely the one satisfying $\sigma_z\geq\sigma_x\geq\sigma_y\geq0$) as state space for the reduced control system and call it the \emph{Weyl chamber}.
In red we show the optimal solution starting from the north pole of the Schmidt sphere and ending at an arbitrary point in the Weyl chamber as obtained from the PMP in Section~\ref{sec:pmp}.
}
\label{fig:weyl}
\end{figure}

Consider a solution $\sigma$ starting at the north pole $\sigma(0)=\sigma_0:=(0,0,1)$ and let $\tau$ denote any desired final state in the Weyl chamber.
Then it holds that $l(0)$ lies in the $xy$-plane.
The only possible solution steering $\sigma$ to $\tau$ which satisfies the PMP and remains within the Weyl chamber is to use $u=(-\tfrac12,\tfrac12,0)$ for time $T_1$ and then to use $u=(0,1,0)$ for time $T_2$ where
$$
T_1=\sqrt2\arccos\big(\sqrt{1-2\tau_y^2}\big), \qquad
T_2=\arccos\bigg(\frac{\tau_x\tau_y+\tau_z\sqrt{1-2\tau_y^2}}{1-\tau_y^2}\bigg).
$$ 
In particular if the final state is $(\tfrac1{\sqrt3},\tfrac1{\sqrt3},\tfrac1{\sqrt3})$ (corresponding to a maximally entangled state), then $T_1=\sqrt2\arccos(\tfrac1{\sqrt3})$ and $T_2=0$.
In the following section we will further investigate this solution.

\subsection{Lifting the solution} \label{sec:lift}

In the final step we derive (approximately) optimal controls for the full control system~\eqref{eq:bilinear} consisting of two coupled qutrits.
For a regular solution $\sigma:[0,T]\to S^2$ this can be done exactly (up to some near instantaneous pulses at the beginning and at the end of the solution) using~\cite[Prop.~2.13]{BipartiteReduced24}.
For non-regular solutions the existence of an exact lift is not guaranteed.
In this section we derive a lifted solution starting at the product state $\ket{\psi_0}=\ket{33}$ and finishing at (or rather arbitrarily close to) the maximally entangled state $\ket{\psi_1}=\tfrac1{\sqrt3}(\ket{11}+\ket{22}+\ket{33})$.


\smallskip In the previous section we have seen that the optimal solution in the reduced control system starting at $(0,0,1)$ and ending at $\tfrac1{\sqrt3}(1,1,1)$ is simply given by a segment of the corresponding great circle traversed at angular velocity 
$\omega^\star/\sqrt2$, where $\omega^\star$ is given in Section~\ref{sec:reduce} and depends on the coupling Hamiltonian $H_0=A\otimes B$.
To simplify the notation we assume (without loss of generality) that the eigenvalues of $A$ and $B$ are $1,0,-1$ and hence $\omega^\star=1$.
Concretely the optimal solution is given by
$$\sigma(t)=\big(\tfrac1{\sqrt2}\sin(\tfrac{1}{\sqrt2}t), \tfrac1{\sqrt2}\sin(\tfrac{1}{\sqrt2}t), \cos(\tfrac{1}{\sqrt2}t)\big),$$
and the maximally entangled state is reached at time $T^\star=\sqrt2\arccos(\tfrac1{\sqrt3})$.
The optimal derivative in the reduced control system is achieved by the generator $-H_U$ using any local unitary $U=V\otimes W$ which satisfies
\begin{equation}
V^*AV = \frac1{\sqrt2}\begin{pmatrix}
0&0&1\\0&0&1\\1&1&0
\end{pmatrix}, \quad
W^*BW = \frac1{\sqrt2}\begin{pmatrix}
0&0&-\iu\\0&0&-\iu\\\iu&\iu&0
\end{pmatrix}.
\end{equation}
Such $U$ exists by assumption on the eigenvalues of $A$ and $B$.
Without applying any controls (other than instantaneously applying $U$ and $U^*$ and the beginning and the end respectively) the evolution of the system is
$$
e^{-\iu U^* A\otimes BUt}\ket{33}
=
-\tfrac{\sin(t)}2 (\ket{11}+\ket{12}+\ket{21}+\ket{22}) + \cos(t)\ket{33},
$$
which, unfortunately, is not a lift of the optimal solution as the singular values are $(\sin(t),0,\cos(t))$.
Indeed, there is no compensating Hamiltonian which yields an exact lift. 
This is because the local control Hamiltonian applied to the system cannot directly affect the derivative of the singular values at time $0$ (see~\cite{Diag22}).
For more background on this issue see~\cite[Sec.~3.1]{Reduced23}.
The origin of the problem is that all states on the solution satisfy $\sigma_x=\sigma_y$, and hence the usual formula for the compensating Hamiltonian does not apply (and might blow-up). 
This will be made clearer in the following.


\smallskip One way to fix this problem is to remain on a path which narrowly avoids the degeneracy $\sigma_x=\sigma_y$ and to compute the corresponding compensating Hamiltonian.
We will consider a solution which remains on the circular path satisfying $\sigma_x=\sigma_y+\varepsilon\sqrt2$ for some small $\varepsilon>0$.
Moreover we use the same generator $-H_U$ with local control unitary $U=V\otimes W$ as above.
As in Section~\ref{sec:opt-ctrl-angle} we can derive the corresponding local compensating Hamiltonian $E\otimes\id+\id\otimes F$.
We obtain 
$$
V^*EV=W^*FW=\frac{\sigma_z}{2\sqrt2\varepsilon} P_y',\quad 
\text{where }
P_y'=
\begin{pmatrix}
0&-\iu&0\\
\iu&0&0\\
0&0&0\\
\end{pmatrix},
$$
which blows up as $\varepsilon$ approaches $0$, but is well-behaved for $\varepsilon\neq0$.
Note also the state dependence via $\sigma_z$.
More precisely, consider the initial state $\tfrac{\varepsilon}{\sqrt2}\ket{11}-\tfrac{\varepsilon}{\sqrt2}\ket{22}+\sqrt{1-\varepsilon^2}\ket{33}$.
Since the reduced solution moves on a circle orthogonal to the $(y-x)$-axis, the reduced solution is given by
$$
\big(\tfrac{\sqrt{1-\varepsilon^2}\sin(t/\sqrt2)+\varepsilon}{\sqrt2}, 
\tfrac{\sqrt{1-\varepsilon^2}\sin(t/\sqrt2)-\varepsilon}{\sqrt2}, 
\sqrt{1-\varepsilon^2}\cos(t/\sqrt2)\big).
$$
Hence, the control Hamiltonian can be written in the time-dependent (instead of state-dependent) way using $V^*EV=W^*FW=\tfrac{\sqrt{1-\varepsilon^2}\cos(t/\sqrt2)}{2\sqrt2\varepsilon}P_y'$.


\smallskip It is also interesting to consider what happens if one applies this control to a solution starting at $U\ket{33}$.
To simplify things we get rid of the time dependence and use the control Hamiltonian given by $V^*EV=W^*FW=\tfrac{1}{2\sqrt2\varepsilon}P_y'$.
Some solutions of this form are shown in Figure~\ref{fig:sols}.
These solutions do not exactly reach a maximally entangled state, but for small $\varepsilon$ they get very close.
Indeed, consider the cost function
$$
C(\varepsilon) = \|\mathrm{sing}(\ket{\psi(T^\star)})-\tfrac{1}{\sqrt3}(1,1,1)\|_2
$$
measuring the Euclidean distance of the reduced state at the final time to the maximally entangled state.
This function is plotted in Figure~\ref{fig:error} and has some interesting properties.
In particular, by trial and error one finds that the function becomes almost exactly periodic when transformed as
$$
\tilde C(x) = \sqrt3 C\big(\tfrac1{2\sqrt2x}\big)(2\sqrt2 x-1).
$$
This allows us to derive a simple approximation for the local minima of $C(\varepsilon)$.
Indeed, the points
\begin{equation} \label{eq:opt-epsila}
\varepsilon_k = (2\sqrt2(x_0+k\Delta x))^{-1}, \quad \text{where } x_0\cong 0.0048,\,\Delta x\cong 2.3252, 
\end{equation}
are close to local minima of $C$ for integer $k$ and close to local maxima for half-integer $k$. 
In practice one should always choose $\varepsilon$ in such a local minimum as it can significantly decrease the final cost.
Indeed, a well chosen $\varepsilon$ achieves a cost similar to that of a badly chosen $\varepsilon$ which is $\frac{1+1/\sqrt3}{1-1/\sqrt3}\cong 3.7$ times smaller.

\begin{figure}[th]
\centering
\includegraphics[width=0.4\textwidth,trim=1cm 5cm 2cm 1.5cm,clip]{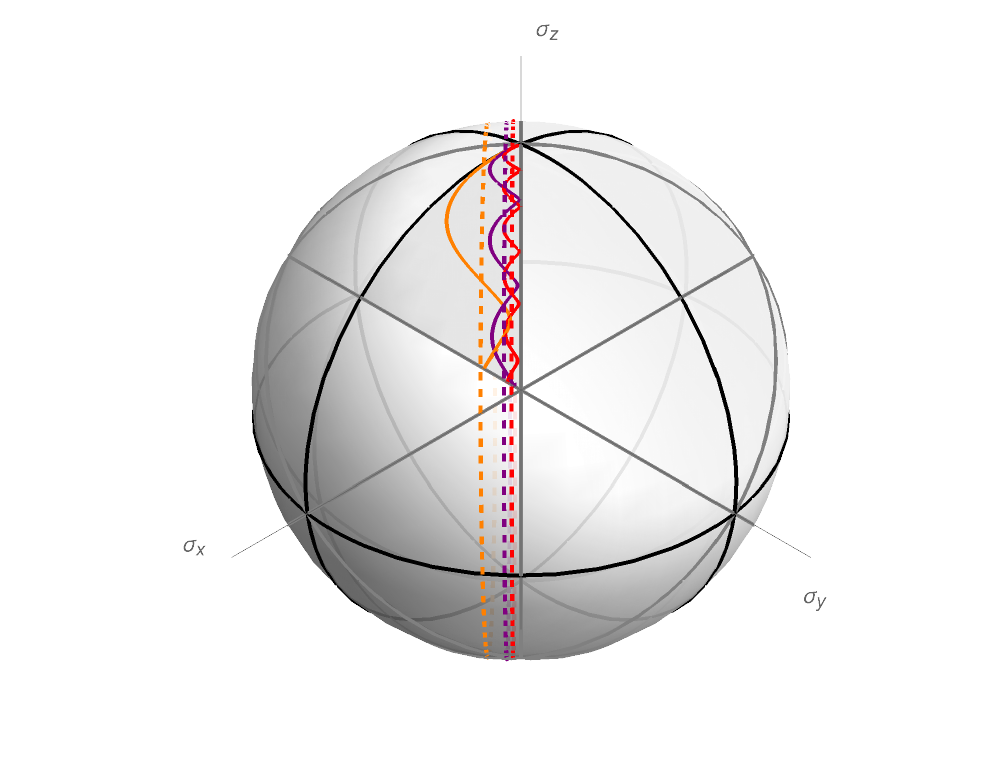}
\includegraphics[width=0.4\textwidth,trim=0cm 3cm 0cm 0cm,clip]{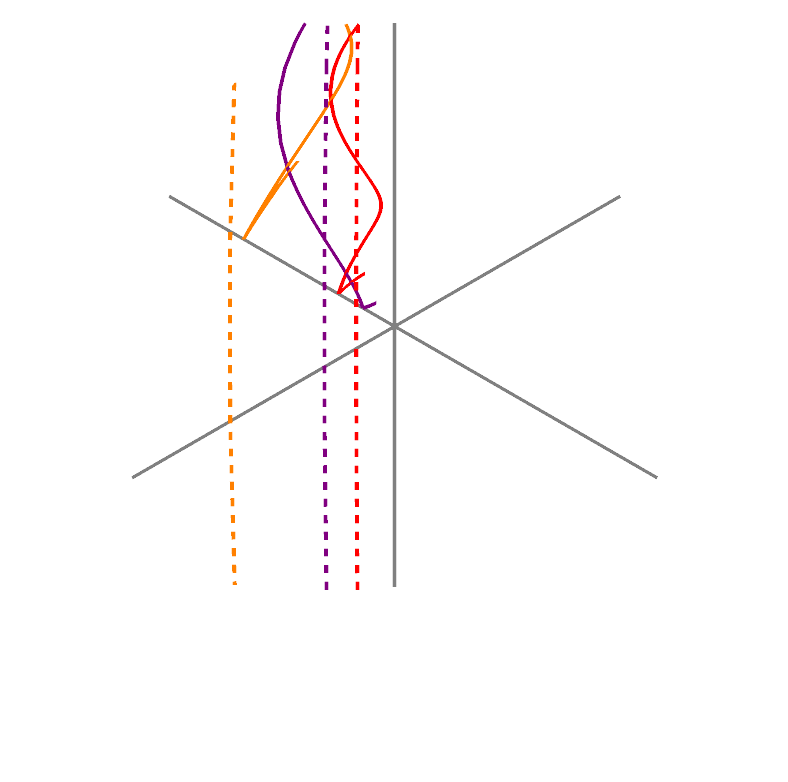}
\caption{Three solutions, with different values of $\varepsilon$, to the optimial control problem projected to the singular values are depicted. 
In each case the solution (solid) is drawn with the circle satisfying $x-y=\varepsilon\sqrt2$ (dashed).
In each case the solution starts at $\ket{33}$ and runs for time $T^\star$ to approximate a maximally entangled state.
The values for $\varepsilon$ chosen are $\varepsilon_1=0.12$ (orange), $\varepsilon_2=0.0506$ (purple), and $\varepsilon_3=0.0276$ (red).
The final distances are $C(\varepsilon_1)=0.140, C(\varepsilon_2)=0.0215,$ and $C(\varepsilon_3)=0.0436$.
Importantly, even though $\varepsilon_2>\varepsilon_3$, the former solution achieves a better result, since the value of $\varepsilon_2$ is chosen such that it achieves a local minimum of the final distance function.
}
\label{fig:sols}
\end{figure}

\begin{figure}
\centering
\includegraphics[width=0.4\textwidth]{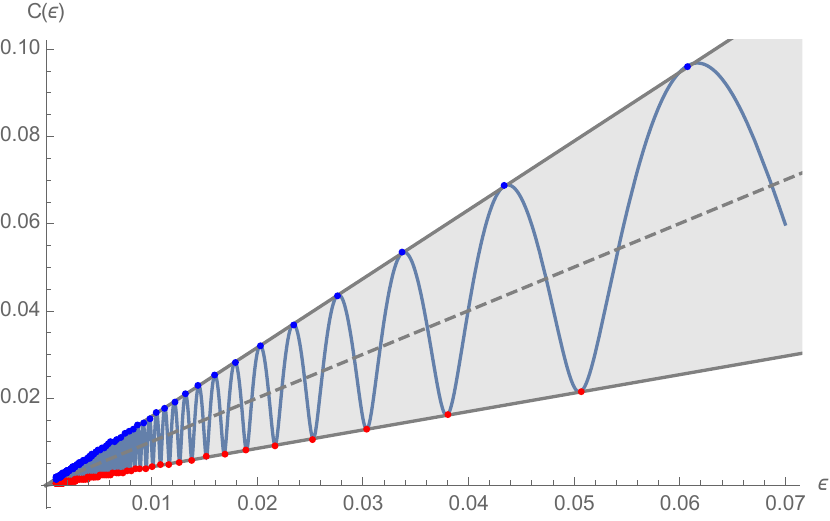}
\includegraphics[width=0.4\textwidth]{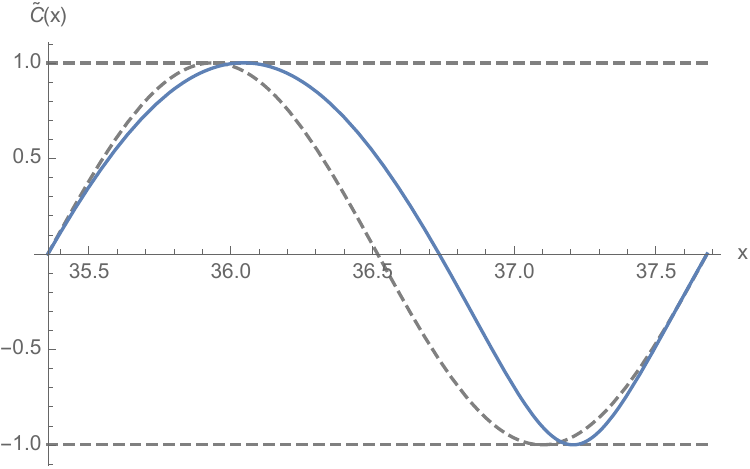}
\caption{
Left: Plot of the cost function $C(\varepsilon)$.
The gray dotted line is $\varepsilon\mapsto\varepsilon$, and the upper and lower gray lines are $\varepsilon\mapsto(1\pm1/\sqrt{3})\varepsilon$.
The (approximate) maxima and minima computed using~\eqref{eq:opt-epsila} are highlighted.
Thus the cost decreases linearly with $\varepsilon$, but, as also shown in Figure~\ref{fig:sols} the exact choice of $\varepsilon$ can make a significant difference.
Right:
The transformed cost $\tilde C(x)$ is almost perfectly periodic (for $x$ large enough and hence $\varepsilon$ small enough). 
Here we depict a single period of $\tilde C(x)$ (blue) alongside a pure sine wave  (gray dashed) for comparison.
}
\label{fig:error}
\end{figure}


\smallskip Finally, once a maximally entangled state is reached, we also wish to stabilize the state.
This can be achieved by switching into the basis where $A,B=\diag(1,0,-1)$.
Using the local compensating Hamiltonian $E\otimes\id+\id\otimes F$ with $V^*EV=W^*FW=\diag(\tfrac12,0,\tfrac12)$ one finds that $\ket{\psi_1}$ is indeed stabilized.

\section{Acknowledgments}

We thank Dominique Sugny for his valuable feedback.
E.M. and L.V.D are supported by the {\em Munich Center for Quantum Science and Technology} (MCQST) and the {\em Munich Quantum Valley} (MQV) with funds from the Agenda Bayern Plus.

\appendix

\section{Decomposition of Drift Hamiltonians} \label{app:ham-decs}

It is most convenient to work with drift Hamiltonians expressed in the form $H_0=\sum_{i=1}^m A_i\otimes B_i$. 
Given an arbitrary drift Hamiltonian $H_0\in\iu\mf{u}(d_1d_2)$ we want to understand the different ways in which it can be written as a sum of decomposable elements $A_i\otimes B_i$. 

It is easy to see that any local term of the form $E\otimes\id$ or $\id\otimes F$ has no effect in the reduced control system. 
Indeed such terms can be determined uniquely:

\begin{remark} \label{rmk:local}
Every Hermitian matrix $H_0\in\iu\mf{u}(d)$ can be written uniquely as $H_0=\tilde H_0+\frac{\tr(H_0)}d\id$ where $\tilde H_0\in\iu\mf{su}(d)$ is a traceless Hermitian matrix.
Indeed, this decomposition is orthogonal with respect to the Hilbert--Schmidt inner product.
For product Hamiltonians $A\otimes B$ we analogously get the unique decomposition into four terms 
$A\otimes B=\tilde A\otimes\tilde B + \frac{\tr(B)}{d_2}A\otimes\id + \frac{\tr(A)}{d_1}\id\otimes B + \frac{\tr(A)\tr(B)}{d_1d_2}\id\otimes\id$.
Since the three latter terms are local, they can be compensated using local unitary control and hence one may for simplicity assume that all Hamiltonians are traceless.
Indeed, while the statements in this section are formulated for general Hermitian matrices, analogous results hold for the traceless case.
This is similar to the unique decomposition of Lindblad generators into coherent and dissipative parts, cf.~\cite{UniqueDecompFvE23}.
\end{remark}

\begin{lemma} \label{lemma:normal-form}
Let $H_0\in\iu\mf{u}(d_1d_2)$ be an arbitrary Hamiltonian and let $A_i$ for $i=1,\ldots,d_1^2$ be a basis of $\iu\mf{u}(d_1)$ and similarly for $B_j$ for $j=1,\ldots,d_2^2$.
Then there is a unique coefficient matrix $C\in\R^{d_1,d_2}$ such that
$$
H_0=\sum_{i,j=1}^{d_1^2,d_2^2} C_{ij} A_i\otimes B_j.
$$
If $A_i=\sum_{k=1}^{d_1^2}S_{ik}\tilde A_k$, $B_j=\sum_{l=1}^{d_2^2}T_{jl}\tilde B_l$ is another choice of bases, with $S\in\R^{d_1,d_1}$ and $T\in\R^{d_2,d_2}$ invertible, then 
$$
H_0=\sum_{k,l=1}^{d_1^2,d_2^2} \tilde C_{kl} \tilde A_k\otimes \tilde B_l \quad \text{where } \tilde C=S^\top C\,T.
$$
As a consequence the rank $r$ of the coefficient matrix is well-defined, i.e.\ it depends only on $H_0$.
Hence the bases $A_i$ and $B_j$ can always be chosen such that
$$
H_0=\sum_{i=1}^r A_i\otimes B_i,
$$
and we will say that this representation is ``in diagonal form''.
\end{lemma}

\begin{proof}
This holds for arbitrary tensor products, cf.~\cite[Thm.~14.7]{Roman05}.
Alternatively this easily follows from the real singular value decomposition and elementary computations.
\end{proof}

\begin{corollary}
Moreover there exist orthonormal bases $A_i$ and $B_j$ (with respect to any given inner product) such that 
$$
H_0=\sum_{i=1}^r \omega_i A_i\otimes B_i,
$$
and we will again say that this representation is ``in diagonal form''. 
The $\omega_i$ are the singular values of the coefficient matrix with respect to any orthonormal basis, and hence they are uniquely defined up to order and sign.
\end{corollary}

\begin{proof}
This follows from Lemma~\ref{lemma:normal-form} and the real singular value decomposition.
\end{proof}

\subsection{Indistinguishable case}

We obtain an analogous result for the case of two indistinguishable subsystems.

\begin{lemma} \label{lemma:normal-form-sym}
Let $H_0\in\iu\mf{u}^s(d^2)$ be an arbitrary coupling Hamiltonian in the indistinguishable case and let $A_i$ for $i=1,\ldots,d^2$ be an orthonormal basis of $\iu\mf{u}(d)$.
Then there is a unique coefficient matrix $C\in\R^{d,d}$ such that
$$
H_0=\sum_{i,j=1}^{d^2} C_{ij} A_i\otimes A_j.
$$
It holds that $C$ is symmetric. 
If $A_i=\sum_{k=1}^{d_1^2}S_{ik}\tilde A_k$ is another choice of basis, with $S\in\R^{d_1,d_1}$ invertible, then 
$$
H_0=\sum_{k,l=1}^{d_1^2,d_2^2} \tilde C_{kl} \tilde A_k\otimes \tilde A_l \quad \text{where } \tilde C=S^\top C\,S.
$$
As a consequence the rank $r$ of the coefficient matrix is well-defined, i.e.\ it depends only on $H_0$.
Hence the basis $A_i$ can always be chosen such that
$$
H_0=\sum_{i=1}^r A_i\otimes A_i,
$$
and we will say that this representation is ``in diagonal form''.
\end{lemma}

\begin{proof}
This follows from the real symmetric eigenvalue decomposition.
\end{proof}

\begin{corollary}
Moreover there exists an orthonormal basis $A_i$ (with respect to any given inner product) such that 
$$
H_0=\sum_{i=1}^{d^2-1} \omega_i A_i\otimes A_i,
$$
and we will say that this representation is ``in diagonal form''.
The $\omega_i$ are the eigenvalues of the coefficient matrix with respect to any orthonormal basis, and hence they are uniquely defined up to order.
\end{corollary}

\begin{remark}
The results of Lemma~\ref{lemma:normal-form} and Lemma~\ref{lemma:normal-form-sym} are in many ways analogous to the normal form results for the Lindblad equation. 
In particular our coefficient matrix $C$ plays the role of the Kossakowski matrix, and by diagonalizing it we obtain a particularly nice form. 
Whereas the eigenvalues of the Kossakowski matrix represent exponential decay rates, the singular values and eigenvalues of our coefficient matrices represent coupling frequencies.
\end{remark}

\printbibliography[heading=bibintoc]

\end{document}